\documentclass[conference,comsoc]{IEEEtran}
\usepackage[T1]{fontenc}
\usepackage{amsmath,amsthm,amssymb,amsmath,bbm, stmaryrd}
\usepackage{tikz,pgfplots}
\usepackage{pgfplotstable}
\pgfplotsset{compat=1.18}
\usetikzlibrary{arrows.meta, decorations.pathreplacing}
\usepgfplotslibrary{statistics,fillbetween}

\usepackage[bookmarks=false]{hyperref}

\usepackage{dblfloatfix}

\usepackage{color}

\usepackage{graphicx} 

\newtheorem{theorem}{Theorem}

\newtheorem{proposition}[theorem]{Proposition}
\newtheorem{lemma}[theorem]{Lemma}

\newcommand{\R}{\mathbb{R}}
\newcommand{\one}{\mathbbm{1}}
\newcommand{\N}{\mathbb{N}}

\newcommand{\setM}{\{1, \dots, M\}}
\newcommand{\setK}{\{1,\dots,K\}}

\renewcommand{\H}{\mathcal{H}}
\newcommand{\E}{\mathcal{E}}
\newcommand{\T}{\mathcal{T}}

\newcommand{\SINR}{\ensuremath{\mathrm{SrINR}}}

\renewcommand{\Pr}{\mathbb{P}}
\newcommand{\Pro}{\Pr^o}

\renewcommand{\epsilon}{\varepsilon}

\title{Performance Guarantees of SIC-Based Decoding for VLSF Codes in Hardcore-Regulated MAC}

\author{
 \IEEEauthorblockN{Kevin Zagalo\IEEEauthorrefmark{1}, Guodong Sun\IEEEauthorrefmark{2}, Philippe Mary\IEEEauthorrefmark{5}, Samir M. Perlaza\IEEEauthorrefmark{2}\IEEEauthorrefmark{3}\IEEEauthorrefmark{4}, Jean-Marie Gorce\IEEEauthorrefmark{1}}

 \IEEEauthorblockA{ Emails: \{kevin.zagalo, guodong.sun, samir.perlaza, jean-marie.gorce\}@inria.fr,  philippe.mary@insa-rennes.fr
             }
   \IEEEauthorblockA{\IEEEauthorrefmark{1} Inria, INSA Lyon, CITI Laboratory, UR3720, Villeurbanne, France
             }
 \IEEEauthorblockA{\IEEEauthorrefmark{2} Centre Inria d’Université Côte d’Azur, Inria, Sophia Antipolis, France
             }
  \IEEEauthorblockA{\IEEEauthorrefmark{5} Univ. Rennes, INSA, CNRS, IETR UMR 6164 F-35000, Rennes, France    }  
     \IEEEauthorblockA{\IEEEauthorrefmark{3} Laboratoire GAATI, Université de la Polyn\'{e}sie fran\c{c}aise, Fa`a`\={a}, French Polynesia
             }
        \IEEEauthorblockA{\IEEEauthorrefmark{4} ECE Dept, Princeton University, Princeton, 08544 NJ, USA
             }

}

\date{\today}

\begin{document}

\maketitle


\begin{abstract}
    Spatial regulation suffices to guarantee reliable, low-latency communication using variable-length stop-feedback (VLSF) codes over multiple access channels. The delay performance of VLSF codes under successive interference cancellation is characterized for transmitters within a coverage region. For deployments modeled by hardcore point processes, the signal-to-interference-plus-noise ratio (SINR) admits a positive lower bound for all transmitters within the coverage region. Under a predefined total error probability constraint, these transmitters are decodable within an upper-bounded delay. For each transmitter, the decoding delay is well approximated by an inverse Gaussian distribution parameterized by the codebook size, target error probability, and the SINR lower bound. Numerical results confirm that hardcore regulation provides uniform performance guarantees across links. 
\end{abstract}

\nopagebreak\section{Introduction}

\noindent Variable-length stop-feedback (VLSF) codes, in which transmission continues until the receiver sends a stop signal, are highly efficient for point-to-point communications in the non-asymptotic regime~\cite{polyanskiy2010variable}. 
In practical networks, these theoretical advantages are realized through hybrid automatic repeat request (HARQ) protocols~\cite{yang2022incremental}. 
However, generalizing these insights to large-scale networks with multiple access channels (MACs) introduces additional challenges.
While variable-length feedback codes for MAC have been analyzed in~\cite{yavas2023variable}, that work abstracts away the spatial randomness of node deployments. 
In practice, spatial configurations induce link-quality variations that make it difficult to guarantee uniform reliability and delay performance across the network.

Stochastic geometry has emerged as a primary framework for modeling large-scale wireless networks, and typically employs Poisson point processes (PPPs) to represent transceiver locations~\cite{haenggi2013stochastic}. 
Despite their analytical tractability, PPP models allow arbitrarily close node placements, resulting in peaky interference that prevents deterministic guarantees on the signal-to-interference-plus-noise ratio (SINR).
To address this issue, recent research has introduced spatial regulation mechanisms that impose geometric constraints on node locations~\cite{ke2024}. 
Specifically, hardcore regulation enforces a minimum separation between interferers to bound interference, while void regulation bounds the maximum distance between transmitter-receiver pairs to ensure a minimum received signal strength. 
These constraints jointly induce a strictly positive lower bound on the SINR across all links, thereby enabling the reliability and latency guarantees required for VLSF-based communication systems.

In this paper, we develop an information-theoretic framework to analyze VLSF codes over a hardcore-regulated MAC, where node locations are modeled by a hardcore point process (HPP).
We consider a decoding strategy based on successive interference cancellation (SIC).
Under this hardcore regulation, the SINR of all transmitters within a coverage region is almost-surely lower-bounded.
By integrating an analysis of block-length distributions with geometry-dependent delay metrics, we characterize reliability-latency guarantees in the multi-transmitter setting. 
The main contributions are summarized as follows:
\begin{itemize}
    \item  We decompose the total error probability into two events: a false-alarm event with probability $\epsilon$, in which the decoder outputs an incorrect message for any transmitter, and a miss-detection event with probability $\epsilon^\prime$, in which at least one transmitter in a $k$-transmitter MAC is not decoded within the delay constraint $t_{\max, k}$ under SIC. Defining the total error probability as $\delta$, we prove that this formulation guarantees $\delta \leq k\epsilon + \epsilon^\prime$ for $k$ transmitters within the coverage region. 
    \item For each link, we approximate the cumulative information density  process induced by the transmitted codeword and the channel output by a diffusion process (DP) with positive drift. The decoding delay is then characterized as the first-passage time of this process crossing a threshold determined by a target error probability $\epsilon$.
    \item Since successful decoding at the $k$-th SIC stage depends on the success of all preceding stages, we show that hardcore regulation enables uniform performance guarantees across links in the coverage region. 
\end{itemize}
The remainder of the paper is organized as follows.
Section~\ref{sec:model} introduces the system model.
Section~\ref{sec:main_result} derives an approximation of the maximal delay for Gaussian codebooks over additive white Gaussian noise (AWGN) channels.
Section~\ref{sec:simulations} presents numerical and Monte Carlo results validating the proposed analysis.
Finally, Section~\ref{sec:conclusion} concludes the paper and discusses possible future work.

\nopagebreak\section{System modeling}\label{sec:model}

\nopagebreak\subsection{System Model}
\noindent Let $\R^2$ denote the Euclidean plane and $b(x,r)$ the open ball centered at $x\in\R^2$ of radius $r > 0$. 
We model transmitter locations by a stationary HPP in $\R^2$, denoted by $\Psi \triangleq \{{\psi_0}, \psi_1, \psi_2, \dots\}$, with hardcore distance $H > 0$, defined on the probability space $(\Omega, \mathcal{A}, \Pr)$. 
This process represents transceivers whose locations are regulated by a hardcore MAC protocol with a minimum separation distance, such that $\|\psi_i - \psi_j\| \geq 2H$ for all $i \neq j$.
We focus on a typical receiver located at $\psi_0$, and treat the remaining points $\Psi \setminus \{\psi_0\}$ as transmitters. 
This setup models the uplink of a representative cell in a cellular network, where nodes outside the coverage region of $\psi_0$ may be associated with other receivers and act as interferers.
Let $o=(0,0)$ denote the origin, and let $\Pro$ denote the Palm probability given that the typical receiver is located at the origin, i.e., $\psi_0=o$. 
The transmitters in $\Psi\setminus\{\psi_0\}$ are ordered by their Euclidean distance from the typical receiver, so that for all $0<i\leq j$, $0 < \|\psi_i-\psi_0\| \leq \|\psi_j-\psi_0\|$.
An event is said to hold $\Pr$-almost-surely ($\Pr$-a.s.) if it occurs with probability one.

In the VLSF framework, communication is dynamically evaluated over an evolving  blocklength $n \in \N$, where  $\N$ is the set of positive integers.
Each transmitter $\psi_{k}$ sends a message $W_k$ via a codeword $X_k^n \triangleq (X_{k,1}(W_k), \dots, X_{k,n}(W_k)) \in \mathbb{R}^n$, and transmission stops at a random stopping time $\tau_k$ (defined later in \eqref{eq:taustardef}). 
The typical receiver employs SIC. Since SIC proceeds sequentially, decoding at the $k$-th stage is possible only after all preceding messages have been correctly decoded.
We therefore define the stage-$k$ SIC completion time 
\begin{equation}\label{def:taustar}
T_{k} \triangleq \max\{\tau_1, \dots, \tau_k\},
\end{equation} which represents the time at which the message of the $k$-th transmitter is decoded. By SIC, the decoded signal of the $k$-th transmitter is canceled at time $T_k$. Thus, the aggregate received signal $Y^n\triangleq(Y_1,\ldots, Y_n)$ at the origin is given by:
\begin{equation}\label{eq:receivedsignal}
    {Y}_n \triangleq \sum_{k=1}^{\infty}  \|\psi_{k} - \psi_0\|^{-\alpha/2} {X}_{k,n} \one_{\{n\leq T_{k} \}} + Z_n, 
\end{equation}
where $\alpha > 2$ is the path-loss exponent and $Z_n\in \mathbb{R}$ is AWGN, and $\forall \omega \in \Omega$, $\forall A \in \mathcal{A}$, $\one_A(\omega) = 1$ if $\omega \in A$, and $\one_A(\omega) =0$ if $\omega \notin A$. 
The messages, codewords in codebooks, and noise sequence are random variables measurable with respect to (w.r.t.) the measurable space $(\Omega, \mathcal{A})$.

We assume that the typical receiver aims to decode signals from a subset of transmitters located within a bounded coverage region $b(\psi_0, R_{\text{cov}})$, where $R_{\text{cov}}$ denotes the radius. 
It is shown in \cite{zagalo2026bis} that for a hardcore-regulated MAC, the number of transmitters within such a region is $\Pr$-almost-surely upper-bounded. 
While we do not reproduce the derivation here, this property is leveraged in Theorem~\ref{thm:SINRbound}.
Specifically, we assume that only the $K\triangleq\Psi(b(\psi_0, R_{\text{cov}}))-1$ transmitters closest to the typical receiver are decoded. 
Note that this boundedness property does not hold for PPPs, as discussed in \cite{ke2024}.
Under this model, signals from transmitters outside the coverage region are treated strictly as interference, which allows the received signal in \eqref{eq:receivedsignal} to be partitioned as $
    {Y}^n = \sum_{k=1}^{K} \one_{\{n\leq T_k\}} \|\psi_{k}\|^{-\alpha/2} {X}^{n}_k + \sum_{k^\prime=K+1}^{\infty} \|\psi_{k^\prime}\|^{-\alpha/2} {X}^{n}_{k^\prime} + Z^n$,
$\Pro$-almost-surely. 
In the second summation, the indicator function is omitted since the transmitters with indices $k^\prime > K$ are not decoded, and their interference persists for all $n$.




\nopagebreak\subsection{VLSF decoding with SIC}\label{sec:VLSF}
\noindent The $K$-transmitter discrete-time memoryless MAC is characterized by a transition probability distribution $P_{Y|X_1, \ldots, X_K}$, obtained by marginalizing over both the Gaussian noise and the non-decodable interference from transmitters with indices $k' > K$.
Given the AWGN nature of the channel, this distribution admits a density $p(y|x_1, \ldots, x_K)$. 
Each transmitter located at $\psi_k$ sends a message $W_k \in \setM$, conveying $\log(M)$ nats of information.
Under a random coding scheme, for each $m \in \setM$, the codeword sequence ${X}^n_k(m)$ is generated with independent and identically distributed (i.i.d.) components drawn from the input distribution $P_{X_k}$. 

\begin{figure*}[th]
\setcounter{equation}{7}
    \begin{equation}\label{eq:mutualpdf}
    \imath(x_k; y | x_1, \dots, x_{k-1}) = \frac{1}{2}\log(1+\eta_k) + \frac{(y - \sum_{j=1}^{k-1}\|\psi_j-\psi_0\|^{-\alpha/2} x_j)^2}{2\sigma_k^2(1+\eta_k)} - \frac{(y - \sum_{j=1}^{k}\|\psi_j-\psi_0\|^{-\alpha/2} x_j)^2}{2\sigma_k^2} 
\end{equation}
\setcounter{equation}{2}
\vspace{.5em}
\hrule
\vspace{-1em}
\end{figure*}

In the SIC framework, decoding at stage $k$ is performed conditioned on the successful decoding of all preceding messages.
We define the correct-history event as:
\begin{equation}\label{eq:correct-history}
    \H_{k-1}\triangleq \bigcap_{i=1}^{k-1} \{ \omega \in \Omega :  \hat{W}_i(\omega) = {W}_i(\omega)\}, \quad \H_0 \triangleq \Omega,
\end{equation}
where $\hat{W}$ denotes the estimated message and will be defined in Section~\ref{sec:stoppingrule}. 
Under the event $\H_{k-1}$, the decoder attempts to decode $W_k$ by evaluating the conditional information density 
\begin{equation}\label{eq:n_letter_information_density}
    \imath_n(x^n_k; {y}^n | {x}_{1}^{n_1}, \dots, {x}^{n_{k-1}}_{k-1}) \triangleq \log \frac{p({y}^n|x^{n_1}_1, \dots, x^{n}_{k})}{p({y}^n|x^{n_1}_1, \dots, x^{n_{k-1}}_{k-1})},
\end{equation}
for any integers $n \geq n_{k-1} \geq \dots \geq n_1$. 
Here, $p({y}^n |  x^{n_1}_1, \dots x^{n}_{k})$ is the marginalized channel density derived from $p({y}^n | x^{n_1}_1, \dots x^{n}_{K})$ by treating the signals from the remaining transmitters $\{k+1, \dots, K\}$ as part of the residual interference. 
Note that if $\mathcal{H}_{k-1}$ does not hold, the decoder proceeds with incorrect interference cancellation, and \eqref{eq:n_letter_information_density} no longer represents the correct log-likelihood ratio. Conditioned on the point process realization $\Psi$ and the correct-history event $\mathcal{H}_{k-1}$ defined in \eqref{eq:correct-history}, each stage of the SIC process can be viewed as an independent point-to-point VLSF decoding problem. 
For each candidate message $m\in \setM$, we define the information density sequence at stage $k$ as:
\begin{equation}\label{eq:def_S_nk}
    S_{k}^{(m)}(n) \!\triangleq\! \imath_n\left(X_k^n(m); Y^n | {X}^{\min\{n, T_{k-1}\}}_{k-1},\!\dots,\!{X}^{\min\{n, T_{1}\}}_1\right).
\end{equation}
By reducing each stage to a point-to-point problem, we can apply the following threshold-based decoding result.
\begin{theorem}[\protect{\cite[Theorem 2.1]{guodong_inria_report}}]\label{thm:threshold_based_rule}
	Let $\epsilon>0$ and $M\in \N$. 
	For any threshold $\gamma\geq\log\frac{M-1}{\epsilon}$, the threshold-based VLSF decoding rule satisfies that the average decoding error probability is upper-bounded by $\epsilon$.
\end{theorem}
While this theorem was derived for VLSF codes over time-correlated fading channels, 
the result was originally established as an achievability bound for variable-length feedback codes over memoryless channels in~\cite[Theorem~3]{polyanskiy2010variable}.


\nopagebreak\subsection{Stopping rule and maximal delay}\label{sec:stoppingrule}

\noindent
Assume all transmitters send messages of size $\log M$ nats with a uniform VLSF target error probability $\epsilon$. 
Let $\gamma \triangleq \log \frac{M-1}{\epsilon}$ be the common decoding threshold. 
At stage $k$, we define the induced stopping time for each message $m\in \setM$ as $\tau_k^{(m)} \triangleq \inf\{n>0 :S_{k}^{(m)}(n) \geq \gamma\}$.
The actual decoding time for the $k$-th transmitter is then \begin{equation}\label{eq:taustardef}
     \tau_k \triangleq \min\{\tau_k^{(1)}, \dots , \tau_k^{(M)} \}.
 \end{equation}
Upon stopping, if $m$ is the unique message such that $\tau_k =\tau_k^{(m)}$, then the decoder returns message $\hat{W}_k \triangleq m$.
Otherwise, the decoder declares an error. 
It holds that $\tau_k \leq  \tau_k^{(m)}$ for any $m\in\setM$ since decoding stops as soon as any codeword, correct or incorrect, reaches the threshold.  Consequently, the stopping time associated with the transmitted message provides an upper bound on the actual decoding time.

Finite decoding times are ensured if, for every considered transmitter, at least one message is eventually decoded, i.e., $\bigcap_{k=1}^K \bigcup_{m=1}^M \big\{\tau_k^{(m)} < \infty\big\}$ occurs with probability one. 
This holds in practical settings with strictly positive SINR, where the information density of the transmitted message has positive drift and thus crosses any finite threshold almost surely~\cite[Eq.~(109)]{polyanskiy2010variable}. Consequently, the maximal delay for each user $k\in\setK$ is finite $\Pr$-almost-surely.

\begin{proposition}\label{prop:pertransmittererror}
    For each $k\in\setK$, let $\epsilon > 0$ and $\delta>0$ such that $\delta - K\epsilon \in (0,1)$, $M\in \N$, $\gamma = \log \frac{M-1}{\epsilon}$.   Let $T_k$ be defined as in \eqref{def:taustar}, and introduce the maximal-delay function
    \begin{equation}\label{eq:def_t_max}
            t_{\max, k}(\epsilon^\prime) \triangleq \inf\left\{n > 0 : \Pro\left(T_k > n \mid \Psi\right) \leq \epsilon^\prime \right\}. 
    \end{equation}
Under SIC with the stated stopping rule, and given $\Psi$, if transmitter $k$ has maximal delay $t_{\max, k}(\delta - k\epsilon)$, then $\Pr\left(\Pro\left(\hat W_k \neq W_k \mid \Psi \right) \leq \delta\right)=1,$ which guarantees that the overall decoding error probability does not exceed $\delta$ almost surely conditioned on $\Psi$.
\end{proposition}

\begin{proof}
    See Appendix~\ref{proof:pertransmittererror}.
\end{proof}

\nopagebreak\section{Deterministic approximation of \texorpdfstring{\ensuremath{t_{\max, k}(\epsilon^\prime)}}{tmax} for Gaussian codebooks and AWGN Channels}\label{sec:main_result}

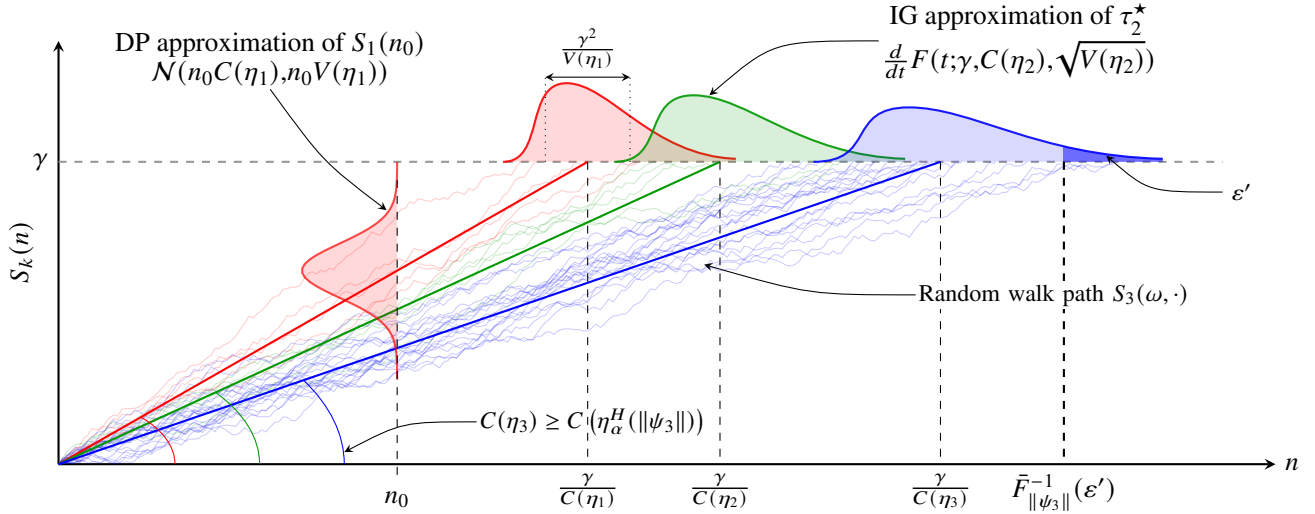
\begin{figure*}[th]
    \centering
    \begin{tikzpicture}[xscale=1.4, yscale=0.8]

    \def\A{5}          
    \def\Cone{1.0}      
    \def\Ctwo{0.8}      
    \def\Cthree{0.6}    
    \def\Tone{\A/\Cone} 
    \def\Ttwo{\A/\Ctwo} 
    \def\Tthree{\A/\Cthree}

    \draw[->,>=stealth, thick] (0,0) -- (11.5,0) node[right] {$n$};
    \draw[->, >=stealth,thick] (0,0) -- (0,7) node[midway, left=0.2cm, rotate=90, anchor=south] {$S_k(n)$};

    \draw[dashed, thick, gray!80] (0,\A) node[left, text=black] {$\gamma$} -- (11,\A);

    \foreach \s/\col/\Cval/\numruns in {1/red/\Cone/5, 6/green!60!black/\Ctwo/5, 11/blue/\Cthree/20} {
        \foreach \run in {1,...,\numruns} {
            \pgfmathsetseed{(\s+\run)*123} 
            \draw[\col, opacity=0.2, thin] (0,0) 
            \pgfextra{
                \def\currX{0}
                \def\currY{0}
            }
            \foreach \i in {1,...,300} {
                \pgfextra{
                    \pgfmathsetmacro{\step}{0.04}
                    \pgfmathsetmacro{\nextYraw}{abs(\currY + \step*\Cval + (0.4*sqrt(\step)*rand))}
                    \pgfmathsetmacro{\nextXraw}{\currX + \step}
                    \ifdim \nextYraw pt > \A pt
                        \pgfmathsetmacro{\fraction}{(\A - \currY) / (\nextYraw - \currY)}
                        \pgfmathsetmacro{\nextX}{\currX + \fraction * \step}
                        \pgfmathsetmacro{\nextY}{\A}
                        \def\dobreak{1}
                    \else
                        \pgfmathsetmacro{\nextX}{\nextXraw}
                        \pgfmathsetmacro{\nextY}{\nextYraw}
                        \def\dobreak{0}
                    \fi
                }
                -- (\nextX, \nextY)
                \pgfextra{
                    \xdef\currX{\nextX}
                    \xdef\currY{\nextY}
                    \ifnum\dobreak=1 \breakforeach \fi
                }
            };
        }
    }

    \draw[red, thick] (0,0) -- (\Tone, \A);
    \draw[green!60!black, thick] (0,0) -- (\Ttwo, \A);
    \draw[blue, thick] (0,0) -- (\Tthree, \A);

    \draw[red] (1.1, 0) arc (0:{atan(\Cone)}:1.1);
    \draw[green!60!black] (1.9, 0) arc (0:{atan(\Ctwo)}:1.9);
    \draw[blue] (2.7, 0) arc (0:{atan(\Cthree)}:2.7);

    \draw[<-, >=stealth, thin] (2.7, 0.2) .. controls (3.25, 0.7) .. (3.95, 0.7) 
        node[right, fill=white, fill opacity=1, inner sep=1pt, font=\small] {$C(\eta_3) \geq C\left(\eta_\alpha^H(\|\psi_3\|)\right)$};

    \foreach \nVal/\amp in {3.2/0.9} {
        \pgfmathsetmacro{\yVal}{\nVal*\Cone}
        \pgfmathsetmacro{\sigma}{0.25*sqrt(\nVal)}
        \begin{scope}[shift={(\nVal, \yVal)}, rotate=90]
            \fill[red, opacity=0.15] plot[domain=-1.8:1.8, samples=50] (\x, {\amp*exp(-(\x*\x)/(2*\sigma^2))}) -- (0,0) -- cycle;
            \draw[red, thick, opacity=0.7] plot[domain=-1.8:1.8, samples=50] (\x, {\amp*exp(-(\x*\x)/(2*\sigma^2))});
            \draw[black, dashed] (\A-3.5, 0) -- (-3.5, 0)  node[below] {$n_0$};
        \end{scope}
    }

    \draw[<-, >=stealth, black, thin] (3.1, 4.2) .. controls (2.5, 5.0) .. (2, 6.2) 
        node[above, font=\Large, fill=white, inner sep=1pt] {\ensuremath{\substack{\text{DP approximation of }S_1(n_0)
\\\mathcal{N}(n_0C(\eta_1), n_0V(\eta_1))}}};

    \foreach \T/\col/\labelIndex/\height/\width in {
        \Tone/red/1/1.3/2.0,
        \Ttwo/green!60!black/2/1.1/2.5,
        \Tthree/blue/3/0.9/3.0%
    } {
        \draw[dashed, black] (\T, \A) -- (\T, 0);
        \node[black, below] at (\T, 0) {$\frac{\gamma}{C(\eta_{\labelIndex})}$};
        
        \fill[\col, opacity=0.15] (\T-\width*0.4, \A)
            .. controls (\T-\width*0.2, \A) and (\T-\width*0.3, \A+\height) .. (\T-\width*0.1, \A+\height)
            .. controls (\T+\width*0.1, \A+\height) and (\T+\width*0.3, \A+0.1) .. (\T+\width*0.7, \A+0.05) 
            -- (\T+\width*0.7, \A) -- (\T-\width*0.4, \A) -- cycle;
            
        \draw[\col, thick, opacity=0.9] (\T-\width*0.4, \A)
            .. controls (\T-\width*0.2, \A) and (\T-\width*0.3, \A+\height) .. (\T-\width*0.1, \A+\height)
            .. controls (\T+\width*0.1, \A+\height) and (\T+\width*0.3, \A+0.1) .. (\T+\width*0.7, \A+0.05);
    }

    \def\qThree{9.5}
    \draw[dashed, black, thick] (\qThree, \A) -- (\qThree, 0) node[below] {$\bar F^{-1}_{\|\psi_3\|}(\epsilon^\prime)$};
    
    \begin{scope}
        \clip (\qThree, \A) rectangle (11.5, 7);
        \fill[blue, opacity=0.5] (\Tthree-3.0*0.4, \A)
            .. controls (\Tthree-3.0*0.2, \A) and (\Tthree-3.0*0.3, \A+0.9) .. (\Tthree-3.0*0.1, \A+0.9)
            .. controls (\Tthree+3.0*0.1, \A+0.9) and (\Tthree+3.0*0.3, \A+0.1) .. (\Tthree+3.0*0.7, \A+0.05) 
            -- (\Tthree+3.0*0.7, \A) -- cycle;
    \end{scope}

    \draw[<-, >=stealth, black, thin] (\qThree+0.4, \A+0.1) .. controls (\qThree+0.8, \A-0.5) .. (\qThree+1.5, \A-0.5)
        node[right, font=\small] {$\epsilon^\prime$};

    \def\vShift{0.6}
    \def\halfWidth{0.4} 
    \draw[<->, >=stealth, black] (\Tone-\halfWidth, \A+\vShift+0.8) -- (\Tone+\halfWidth, \A+\vShift+0.8)
        node[midway, above, font=\small] {$\frac{\gamma^2}{V(\eta_1)}$};
    \draw[dotted, black] (\Tone-\halfWidth, \A) -- (\Tone-\halfWidth, \A+\vShift+0.8);
    \draw[dotted, black] (\Tone+\halfWidth, \A) -- (\Tone+\halfWidth, \A+\vShift+0.8);

    \draw[<-, >=stealth, black, thin] (\Ttwo-0.1, \A+1.1) .. controls (\Ttwo+0.5, \A+2.0) .. (\Ttwo+1.5, \A+2.0) 
        node[right, font=\Large, inner sep=1pt] {$\substack{\textrm{IG approximation of } \tau^\star_2 \\ \frac{d}{dt}F(t;\gamma, C(\eta_2), \sqrt{V(\eta_2)})}$};

    \draw[<-, >=stealth, black, thin] (6.1, 3.3) .. controls (7.1, 2.8) .. (8.1, 2.8) 
        node[right, font=\small, fill=white, inner sep=1pt] {Random walk path $S_3(\omega, \cdot)$};
\end{tikzpicture}
    \caption{Illustration of the DP and IG approximations for $K = 3$.}
    \label{fig:approximation}
    \vspace{-0.1cm}
\end{figure*}

\noindent  Since $t_{\max, k}(\epsilon^\prime)$ depends on the locations $\Psi$, it is a random variable. 
In this section, we first use the properties of HPPs in order to provide a deterministic upper-bound on $t_{\max, k}(\epsilon^\prime)$. 
Second, we approximate this bound by using the memoryless property of the channel. 

\nopagebreak\subsection{Codebook and AWGN channel}

\noindent Let the codeword associated with each transmitter $k$ and each message $m$ be generated according to a zero-mean Gaussian distribution\footnote{While Gaussian signaling can be suboptimal in the non-asymptotic regime, we select it here for its analytical tractability.}, i.e., $P_{X_k} = \mathcal{N}(0, \mathsf{P})$, where $\mathsf{P} > 0$ denotes the average transmit power per symbol. Consequently, under the path-loss model \eqref{eq:receivedsignal}, the signal received from transmitter $k$ is distributed as $\mathcal{N}(0, \|\psi_k-\psi_0\|^{-\alpha}\mathsf{P})$. The AWGN power is $\sigma^2 > 0$.
The marginalization over the non-decoded signals, i.e., treating interference as noise, effectively increases the noise power for the $k$-th stage of SIC. The resulting residual-interference-plus-noise power is thus 
$\sigma_k^2 = \sigma^2 + \mathsf{P}\sum_{j=k+1}^\infty \|\psi_j-\psi_0\|^{-\alpha}$. 
Let the signal-to-residual-interference-plus-noise ratio (\SINR) at the $k$-th decoding stage be defined as $\eta_k \triangleq \frac{\mathsf{P}\|\psi_k-\psi_0\|^{-\alpha}}{\sigma_k^2}$. 
Consequently, for Gaussian codebooks over an AWGN channel, the single-symbol conditional information density at stage $k$, denoted by $\imath(x_k; y | x_1, \dots, x_{k-1})$, is obtained by substituting the probability density function of the Gaussian distribution into~\eqref{eq:n_letter_information_density}, yielding  \eqref{eq:mutualpdf}.

\nopagebreak\subsection{Hardcore Regulation}

    
     \noindent By definition, HPPs are \textit{locally bounded}. Specifically, for $R \geq 0$, the number of points within a ball $b(o,R)$ satisfies
     \setcounter{equation}{8}
\begin{equation}\label{eq:spaceregulation}
     \Psi(b(o,R)) \leq 1 + \rho_H R + \nu_H R^2, \quad \Pro-a.s.,
     \end{equation}
     where $\rho_H \triangleq \frac{\pi}{\sqrt{3}H}$ and $\nu_H \triangleq \frac{\pi}{\sqrt{12}H^2}$, cf. \cite[Lemma 3]{ke2024}. Then, the number of transmitters $K$ within the coverage region $b(\psi_0, R_{\text{cov}})$ is $\Pro$-almost-surely upper-bounded by $K \leq \rho_H R_{\text{cov}} + \nu_H R_{\text{cov}}^2$.

\begin{theorem}[\protect{\cite[Remark~11]{zagalo2026bis}}]
    \label{thm:SINRbound}
    Let $\Psi$ be a HPP with hardcore distance $H > 0$ and path-loss exponent $\alpha > 2$. Define the function $\eta^H_\alpha(r)$ as \begin{equation}\label{eq:sinrbound}
    \eta^H_\alpha(r) \triangleq \left(\frac{\pi^2}{\sqrt{3}}  + \frac{\pi \rho_H}{(\alpha-1)}r +  \frac{\pi \nu_H }{(\alpha-2)}r^2 + \frac{\sigma^2}{\mathsf{P}}r^{\alpha} \right)^{-1}.\end{equation} 
    Then, $\forall k \in \setK$, $\eta_k \geq \eta^H_\alpha(\|\psi_k\|)$, $\Pro$-almost surely.
\end{theorem}

Thanks to hardcore regulation, $t_{\max, k}(\epsilon^\prime)$ is not only  finite for any $\epsilon^\prime > 0$, but also upper-bounded. As illustrated in \figurename~\ref{fig:approximation}, if the slope $C(\eta_k)$ of the cumulative information density is lower-bounded by a strictly positive number, then $t_{\max, k}(\epsilon^\prime)$ 
is $\Pro$-almost-surely upper-bounded. In the following, we derive an approximation of $t_{\max, k}(\epsilon^\prime)$.

\nopagebreak\subsection{Information as a random walk}

\noindent Thanks to the memoryless property of the channel and the application of the random coding scheme, the cumulative information associated with the transmitted message $W_k$ at time $n\in\N$ can be modeled as a random walk.
This process is defined as the sum $S_k(n) = \sum_{j=1}^n I_{k,j}$, where the information increment at the $j$-th channel use is given by: $$I_{k,j} = \imath\left(X_{k,j}; Y_j \mid X_{1,j}, \dots, X_{k-1,j}\right).$$
Under the probability $\Pro_k : A\in\mathcal{A} \to \Pro(A \mid \Psi, W_k, \H_{k-1})$, the conditional information density increments are i.i.d.
Specifically, the expectation of the increments is the \textit{mutual information} $C(\eta_k) \triangleq \frac{1}{2}\log(1+\eta_k)$, so that $\mathbb{E}^o_k\left[I_{k,j}\right] = C(\eta_k)$,\footnote{\label{footnote:Po} The operators $\mathbb{E}^o_k$ and $\mathbb{V}^o_k$ are respectively the expectation and variance associated with the Palm probability $\Pro_k$.} $\Pro_k$-almost-surely.
Their variance corresponds to the \textit{channel dispersion} $V(\eta_k) \triangleq \frac{\eta_k}{\eta_k+1}$, so that $\mathbb{V}^o_k\left[I_{k,j}\right] = V(\eta_k)$,\textsuperscript{\ref{footnote:Po}}$ \Pro_k$-almost-surely. 
These quantities align with the standard results for Gaussian inputs over additive Gaussian channels, as detailed in \cite[Eqs. (2.13), (2.56)]{molavianjazi2014unified}. 
Consequently, the cumulative information density process $S_k \triangleq \left(S_k(n), n \in \mathbb{N}\right)$ behaves as a random walk with drift $C(\eta_k)$ and diffusion coefficient $\sqrt{V(\eta_k)}$. 
This stochastic characterization, which was originally introduced in \cite{polyanskiy2010variable}, facilitates the analysis of the decoding delay distribution and is illustrated in \figurename~\ref{fig:approximation}.


Recall that the actual decoding delay for transmitter $k$, $\tau_k$, as defined in \eqref{eq:taustardef}, is upper-bounded by the delay induced by the transmitted codeword. This delay is characterized as the first-passage time of the above random walk:
    \begin{equation}\label{eq:delaydef}
        \tau_k^\star = \inf\left\{ n \in \N : S_k(n) \geq \gamma \right\}.
    \end{equation}
In the following, we use the asymptotic properties of random walks to approximate the distribution of $\tau_k^\star$.

\nopagebreak\subsection{Diffusion Process approximation}
\noindent A common approximation associated with random walks comes from Donsker's theorem \cite{donsker1951invariance}, which states that the rescaled process $\left(\frac{S_k(\lfloor Nt \rfloor) - \lfloor Nt \rfloor C(\eta_k)}{\sqrt{N V(\eta_k)}}, t \geq 0 \right)$ converges weakly to a standard Brownian motion $B$ defined on $(\Omega, \mathcal{A}, \Pr)$ when $N\rightarrow \infty$. 
In our case, we aim to approximate the distribution of $\tau_k^\star$, namely the first-passage time of the random walk $S_k$ through a given threshold $\gamma$, which is discrete, by the first-passage time of a DP through $\gamma$, which is continuous. Below, $n$ denotes the discrete time index of the random walk $S_k(n)$, whereas $t$ denotes the continuous time variable of its diffusion approximation.
In the context of this paper, the DP approximation of $S_k$ leads to the DP $\hat{S}_k(t) = tC(\eta_k) +  \sqrt{V(\eta_k)} B(t)$, which implies that for all $n\in\N$, $S_k(n)$ is approximated by a Gaussian distribution $\mathcal{N}(nC(\eta_k), nV(\eta_k))$. 
This approximation, illustrated in \figurename~\ref{fig:approximation}, is accurate in the regime $\gamma \gg C(\eta)$, i.e. when the expected number of channel uses to threshold is large.  
The DP approximation provides a continuous approximation of the first-passage time $\tau_k^\star$ as an inverse Gaussian (IG) variable; see \cite[Section~3.2.2]{jeanblanc2009mathematical}. 
For a DP with positive drift $\mu$ and diffusion standard deviation rate $\nu$, the cumulative distribution function (CDF) of the first passage time to a threshold $\gamma > 0$ is, for $t > 0$, given by the IG CDF
\begin{multline}\label{eq:IGCDF}
F(t; \gamma, \mu, \nu) = \Phi \left( \frac{\mu t - \gamma}{\nu\sqrt{t}} \right) + e^{  \frac{2\gamma\mu}{\nu^2}} \Phi \left( \frac{-\mu t - \gamma}{\nu \sqrt{t}} \right),
\end{multline} 
where $\Phi$ is the CDF of the standard normal distribution.

Let $\hat{\tau}_k$ be the first-passage time of the approximating process $\hat{S}_k$ through the information threshold $\gamma$. 
For VLSF codes, the DP approximation leads to the approximation of the distribution of $\tau_k^\star$ under $\Pro_k$, by the CDF of $\hat{\tau}_k$, which follows an IG distribution with mean $\frac{\gamma}{C(\eta_k)}$ and shape $\frac{\gamma^2}{V(\eta_k)}$, cf. \cite[Section~3.2.2]{jeanblanc2009mathematical},
\begin{equation}\label{eq:DPapprox}
   \Pro_k\left(\hat{\tau}_k \leq n\right) = F\left(n; \gamma, C(\eta_k), \sqrt{V(\eta_k)}\right).
\end{equation}

We have the following maximal-delay bound.

\begin{proposition}\label{prop:approxtmax}
    Let $\Psi$, $H > 0$, $\alpha > 2$, and $\eta^H_\alpha$ be defined as in Theorem~\ref{thm:SINRbound}. Let $k\in\setK$, $\epsilon^\prime > 0$, and let $\hat{t}_{\max, k}(\epsilon^\prime)$ be as defined in \eqref{eq:def_t_max} for the first-passage time $\hat{\tau}_{k}$. Let $c^H_\alpha(r) =C(\eta^H_\alpha(r))$, $v^H_\alpha(r) =V(\eta^H_\alpha(r))$, with $\eta^H_\alpha$ defined in \eqref{eq:sinrbound}, and $\bar F_r(t) = 1 - F\left(t; \gamma, c^H_\alpha(r), \sqrt{v^H_\alpha(r)}\right)$, $0< r \leq R_{\text{cov}}$,  so that $F_{\|\psi_k\|}$ is a stochastic upper bound on $\hat{\tau}_k$. Then, $\Pro_k$-almost-surely, \begin{equation*}
        \hat{t}_{\max, k}(\delta - k\epsilon) \leq \bar F^{-1}_{\|\psi_k\|}\left(\frac{\delta}{\rho_H\|\psi_k\| + \nu_H \|\psi_k\|^2} - \epsilon \right).
    \end{equation*} 
    Note that $\bar F_r^{-1}$ is the inverse tail function of an IG distribution, which can be computed explicitly.
\end{proposition}

\begin{proof}
    See Appendix~\ref{proof:approxtmax}.
\end{proof}


\begin{figure}[t]
    \centering
    \input{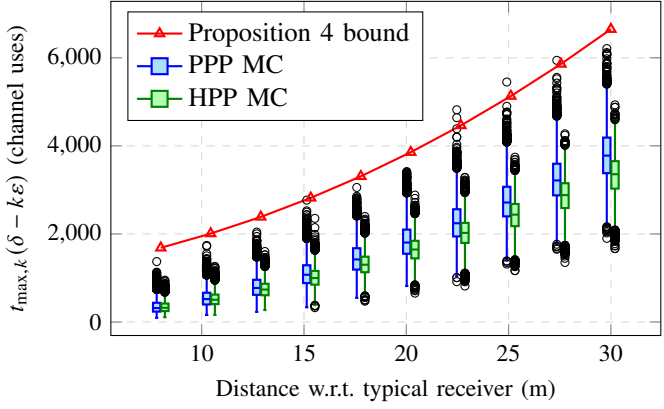}
    \caption{$\delta = 10^{-3}, \epsilon=10^{-5}, \alpha=4, H=4, M=2^{50}, \sigma^2/\mathsf{P} = -100$ dB. }
    \label{fig:tmax}
\end{figure}

\nopagebreak\section{Numerical Results}\label{sec:simulations}
\noindent The simulated HPPs (green boxplot) in \figurename~\ref{fig:tmax} are Matérn hardcore processes of type II \cite{haenggi2013stochastic}, generated by thinning a stationary PPP ($\lambda = 1$) to enforce a minimum inter-point distance of $2H=8$. This yields an intensity $\lambda_H = \frac{1-e^{-4\lambda H^2\pi}}{4H^2\pi}$ \cite[p. 58]{haenggi2013stochastic}. For a fair comparison, the reference PPPs (blue boxplot) are simulated with the same intensity $\lambda_H$.
Box-and-whisker plots display the simulated delay statistics (min, max, and quantiles) over 1\,000 Monte Carlo runs versus distance for 1\,000 simulated point processes (PPP and HPP). The theoretical bound from Proposition~\ref{prop:approxtmax} is plotted as the red curve, which illustrates that, to guarantee a decoding probability of at least $1- \delta$, a transmitter at distance $r > 0$ must be allocated a maximal delay of $\bar F^{-1}_r\left(\frac{\delta}{\rho_H r + \nu_H r^2} - \epsilon\right)$. In particular, the latency statistics shown in \figurename~\ref{fig:tmax}  demonstrate that without hardcore regulation, the PPP delays can exceed this theoretical bound, whereas the HPP latencies remain strictly below it. This implies that unregulated transmitter locations, which may be arbitrarily close together, prevent maximal delay guarantees when VLSF decoding is coupled with SIC. However, employing a hardcore-regulated MAC resolves this issue, ensuring a maximum delay for transmitters within a given area at a target error probability.



\nopagebreak\section{Conclusion}\label{sec:conclusion}

\noindent 
We developed a framework for analyzing the reliability and latency of VLSF codes over hardcore-regulated MAC. 
Using an inverse Gaussian approximation for the decoding delay, we obtained a maximal-delay bound that depend on the transmitter-receiver distance, the target reliability, and the number of SIC-decoded users.  
Since transmitters are indexed by increasing distance, this bound grows with the decoding stage, implying that the latency budget is dictated by the transmitters decoded last.
Numerical results confirm that hardcore regulation is suffice for uniform performance guarantees. 
Future work will extend this analysis to fading channels and incorporate void regulation to jointly control signal strength and interference.

\begin{appendices}

\section{Proof of Proposition~\ref{prop:pertransmittererror}}\label{proof:pertransmittererror}
\noindent Define the stage-wise error event occurred at decoding the message of transmitter $k$ as $\E_k(m) \triangleq  \{\hat{W}_k \neq m, {W}_k = m\}$, which is characterized by the conditional information density process $S_{k}^{(m)}(n) $ in \eqref{eq:def_S_nk}. 
Hence, under SIC, the error event for transmitter $k$ can be partitioned as \vspace{-0.4em}\begin{equation}\label{eq:partition} 
\{\hat W_k \neq W_k \} = \bigcup_{i=1}^k\E_i \cap \H_{i-1},
\end{equation} where $\E_i$ denotes the per-stage error event $\E_i \triangleq \bigcup_{m=1}^M\E_i(m) = \{\hat W_i \neq W_i\}$. The goal of the proof is to show that taking the threshold $\gamma = \log\frac{M-1}{\epsilon}$, coupled with the 
maximal delay $t_{\max, k}(\delta-k\epsilon)$, ensures that  $\Pr\left(\hat W_k \neq W_k \mid \Psi \right) \leq \delta$, almost-surely, for $k \in \setK$. This conditional error event can be decomposed into two types of errors: 
\paragraph{False alarm}
For $W_k=m$, given all previous stages are decoded correctly, i.e., under $\mathcal{H}_{k-1}$, an incorrect codeword from transmitter $k$ crosses the threshold at time $n$ if the event $\E_{k,n}^{\rm FA}(m) \triangleq \left(\bigcup_{m^\prime \neq m} \left\{S_{k}^{(m^\prime )}(n) \geq \gamma \right\} \right)  \cap \left( \bigcap_{p=0}^{n-1} \bigcap_{m^\prime \in \setM} \left\{S_{k}^{(m^\prime)}(p) < \gamma \right\}\right)$ occurs.
The false alarm event at stage $k$, when $W_k=m$, is 
\begin{equation}\label{eq:defFA}
\E_{k}^{\rm FA}(m) \triangleq \bigcup_{n=1}^\infty \E_{k,n}^{\rm FA}(m) \cap \{{\tau}_k = n \}.
\end{equation}
For VLSF codes for the point-to-point channel, Theorem~\ref{thm:threshold_based_rule} implies that, for $\gamma \geq \log \frac{M-1}{\epsilon}$, $\Pro$-almost-surely, \begin{equation}\label{eq:falsealarmprobability}
    \Pro\left(\E_{k}^{\rm FA}(m) \mid W_k=m, \H_{k-1}, \Psi\right) \leq \epsilon.
\end{equation}




\paragraph{Miss detection}
   
   No codeword from transmitter $k$ crosses the threshold before time $n$, i.e., $\E_{k,n}^{\rm MD} \triangleq \bigcap_{p=1}^n \bigcap_{m'\in\setM} \left\{ S_{k}^{(m')}(p) < \gamma \right\}.$ 
Under SIC, the decoding process for transmitter $k$ stops with probability $1-\epsilon^\prime$ before the time instant $t^\prime \triangleq \inf\left\{n > 0 : \Pro\left(\bigcup_{i=1}^k \E_{i,n}^{\rm MD} \cap \H_{i-1} \mid \Psi \right) \leq \epsilon^\prime \right\}$.

The conditional missed detection event is equivalent to the event where the stopping time exceeds the deadline, i.e., $\E_{i,n}^{\rm MD} \cap \H_{i-1} = \{{\tau}_i > n\}$, and the miss detection probability using SIC of transmitter $k$ is $\bigcup_{i=1}^{k} \E_{i,n}^{\rm MD} \cap \H_{i-1}  = \{T_k > n\}$. Thus, $t_{\max, k}(\epsilon^\prime) = t^\prime$, $\Pro$-almost-surely.

\paragraph{Total error}
By the definition of VLSF false alarm and miss detection, if the maximal delay is set to $n^\prime$, the error event satisfies $ \E_i(m) \subseteq \left(\bigcup_{n=1}^{n^\prime}   \E_{i,n}^{\rm FA}(m) \cap \{{\tau}_i = n \} \right) \cup  \E_{i,n^\prime}^{\rm MD}$, hence 
\begin{equation}\label{eq:boundform}
    \E_i(m) \overset{\eqref{eq:defFA}}{\subseteq} \E_{i}^{\rm FA}(m)  \cup  \E_{i,n^\prime}^{\rm MD}
\end{equation}

Combining \eqref{eq:partition} and \eqref{eq:boundform}, for $n^\prime = t_{\max,k}(\epsilon')$, we obtain
\begin{multline}\label{eq:unionbound}
    {\small \bigcup_{i=1}^k \E_i \subset\left(\bigcup_{i=1}^{k}   \bigcup_{m=1}^M \E_{i}^{\rm FA}(m)\right)  \cup \left(\bigcup_{i=1}^{k}  \E_{i,t_{\max,k}(\epsilon')}^{\rm MD} \cap \H_{i-1}\right).}
\end{multline}
Thus, since the messages $m$ are equally likely, applying the union bound yields, $\Pro$-almost-surely,
\begin{align}
   &\Pro\left(\hat W_k \neq W_k \mid \Psi\right) \overset{\eqref{eq:partition}}{=} \Pro\left(\bigcup_{i=1}^{k} \E_i \cap \H_{i-1}\mid \Psi \right) \nonumber\\
    &\overset{\eqref{eq:unionbound}}{\leq} \sum_{i=1}^k \frac{1}{M}\sum_{m=1}^M\Pro\left(\E_{i}^{\rm FA}(m) \cap \H_{i-1} \mid \Psi, W_i=m\right) \nonumber\\& \quad \quad \quad \quad + \Pro\left(\bigcup_{i=1}^{k} \E_{i,{t_{\max,k}(\epsilon')}}^{\rm MD}\cap \H_{i-1}\mid \Psi \right) 
    \leq k\epsilon + \epsilon^\prime \label{eq:bayes} 
\end{align}
where \eqref{eq:bayes} comes from \eqref{eq:falsealarmprobability} and the inequality $\Pro(A \cap B \mid C) = \Pro(B\mid C)\Pro(A \mid B, C) \leq \Pro(A \mid B, C)$, for all $A,B,C\in\mathcal{A}$. Finally, setting $\epsilon^\prime = \delta - k\epsilon$ completes the proof. \hfill $\square$


\section{Proof of Proposition~\ref{prop:approxtmax}}\label{proof:approxtmax}

\noindent We first prove that the first-passage time is stochastically decreasing w.r.t. the \SINR.

\begin{lemma}\label{lem:stochasticdom}
    Let $\gamma>0$ be the decoding threshold. If $\eta_1 \geq \eta_2$, then $\forall t > 0$, the IG cumulative distribution function satisfies $F\left(t; \gamma, C(\eta_1),\sqrt{V(\eta_1)}\right) \geq F\left(t; \gamma, C(\eta_2),\sqrt{V(\eta_2)}\right).$
\end{lemma}

\begin{proof}
Let $F(t)$ be the IG CDF as defined in \eqref{eq:IGCDF}. To prove stochastic dominance induced by a higher SrINR for all $t > 0$, it suffices to show that the CDF is monotonically increasing with respect to the SrINR, i.e.,
$\frac{dF(t)}{d\eta} > 0.$
We rewrite the IG CDF as $F(t) = \Phi(A) + e^B\Phi(C),$
where $A = \frac{\mu t - \gamma}{\nu\sqrt{t}}, B = \frac{2\gamma\mu}{\nu^2}, C = \frac{-\mu t - \gamma}{\nu\sqrt{t}}$. Here, $\mu = \frac{1}{2}\log(1+\eta)$ and $\nu = \sqrt{\frac{\eta}{1+\eta}}$.
Let $\phi = \Phi^\prime$. We have 
$\frac{dF}{d\eta}
= \phi(A) \frac{dA}{d\eta}
+ e^B\Phi(C)\frac{dB}{d\eta}
+ e^B\phi(C)\frac{dC}{d\eta}.$
For the Gaussian density function $\phi$, it holds that $e^B\phi(C) = \phi(A)$, so
\begin{equation}\label{eq:chain_rule_IG_CDF}
\frac{dF}{d\eta}
= \phi(A)\left( \frac{dA}{d\eta} + \frac{dC}{d\eta} \right)
+ e^B\Phi(C)\frac{dB}{d\eta}.
\end{equation}

The first component is 
\begin{equation}\label{eq:simplification1}
\frac{dA}{d\eta} + \frac{dC}{d\eta}
= \frac{d}{d\eta}\left( \frac{-2\gamma}{\nu\sqrt{t}} \right) =  \frac{2\gamma}{\nu^2\sqrt{t}} \frac{d\nu}{d\eta} > 0,
\end{equation}
where \eqref{eq:simplification1} comes from $\frac{\mu t - \gamma}{\nu\sqrt{t}} + \frac{-\mu t - \gamma}{\nu\sqrt{t}} = \frac{-2\gamma}{\nu\sqrt{t}}$, $\gamma, \nu,t > 0$ and $d\nu/ d\eta > 0$. 
Next, we analyze $\frac{dB}{d\eta} = 2\gamma \frac{d}{d\eta} (\frac{\mu}{\nu^2})$.
Define $h(\eta) = \frac{\mu}{\nu^2}
= \frac{(1+\eta)\log(1+\eta)}{2\eta}$. Let $x = 1+\eta > 1$. Then $h(x) = \frac{x \log x}{2(x-1)}$,
and its derivative is $\frac{dh(x)}{dx}= \frac{1}{2(x-1)^2} \left( x - 1 - \log x \right).$
Since $\log x < x - 1$ for all $x>1$, it follows that $\frac{dh}{dx} > 0$, and hence $\frac{dB}{d\eta} > 0$. Finally, since $\phi(A) > 0$, $\Phi(C) > 0$  and $e^B>0$, both terms in \eqref{eq:chain_rule_IG_CDF} are strictly positive. Thus, $\frac{dF(t)}{d\eta} > 0 \quad \forall t>0$, which concludes the proof.
\end{proof}

\begin{proof}[End of proof of Proposition~\ref{prop:approxtmax}]
The CDF of $\max_{i\leq k}\tau_i^\star$, under $\Pro_k$, is approximated with \eqref{eq:DPapprox} by $G_k(t) \triangleq \prod_{i=1}^k F\left(t; \gamma, C(\eta_i), \sqrt{V(\eta_i)}\right)$.
Let $\epsilon^\prime = \delta - k\epsilon$. 
As an approximation of $t_{\max, k}(\epsilon^\prime)$, we solve $t^\prime = \inf\{ t > 0 : G_k(t) \geq 1-\epsilon^\prime\}.$ According to Lemma~\ref{lem:stochasticdom}, $F(t; \gamma, c^{H}_\alpha(r), \sqrt{v^{H}_\alpha(r)})$ decreases w.r.t. $r>0$ since $\eta_\alpha^H$ is a decreasing function of $r$. The latter coupled with Theorem~\ref{thm:SINRbound}, i.e. $\eta_k \geq \eta_\alpha^H(\|\psi_k\|)$, yields that it is sufficient to solve $$t^\prime \leq \inf\{ t > 0 : \bar F_{\|\psi_k\|}(t) \leq 1- (1 - \epsilon^\prime)^{1/k}\},$$ in order to satisfy the bound on $\epsilon^\prime$. Thus, with \eqref{eq:spaceregulation} and the fact that $\bar F^{-1}_r(\epsilon^\prime)$ decreases w.r.t. $\epsilon^\prime$, $t^\prime$ is upper-bounded by $\bar F_r^{-1}(1-(1-\epsilon^\prime)^{1/k})$. Furthermore, since $0 < \frac{1}{k} \leq 1$ and $-\epsilon^\prime \geq -1$, with Bernoulli's inequality we get $1-(1-\epsilon^\prime)^{1/k} \geq \epsilon^\prime /k.$ Thus $\bar F^{-1}_r(1-(1-\epsilon^\prime)^{1/k}) \leq \bar F^{-1}_r(\epsilon^\prime/k).$ Using \eqref{eq:spaceregulation}, we get $k \leq \rho_H \|\psi_k\| + \nu_H \|\psi_k\|^2$, thus $\frac{\delta - k\epsilon }{\rho_H\|\psi_k\| + \nu_H \|\psi_k\|^2} \geq \frac{\delta}{\rho_H\|\psi_k\| + \nu_H \|\psi_k\|^2} - \epsilon$. Since $\bar F^{-1}_r$ is decreasing, $\lim_{N\to\infty} \hat{t}_{\max, k}(\delta - k\epsilon) \leq \bar F^{-1}_{\|\psi_k\|}\left(\frac{\delta}{\rho_H\|\psi_k\| + \nu_H \|\psi_k\|^2} - \epsilon\right)$.
    \end{proof}

\end{appendices}

\section*{Acknowledgment}
   \noindent This work is supported in part by the French National Agency for Research (ANR) through the project ANR-23-CMAS-0023 of the RIS3 program and the project ANR-22-PEFT-0010 of the France 2030 program PEPR Réseaux du Futur; and in part by the Agence de l'innovation de défense (AID) through the project UK-FR 2024352.

\bibliographystyle{IEEEbib}
\bibliography{references}

\end{document}